\documentclass[11pt]{article}
\usepackage[margin=1in]{geometry}

\usepackage{graphicx}
\usepackage{caption}
\usepackage{float}
\usepackage{xcolor}
\usepackage{graphics}
\usepackage{cite}
\usepackage{caption}
\usepackage{bm}
\usepackage{epstopdf}
\usepackage{amsmath}
	\usepackage{amssymb}
	 \usepackage{enumerate}
	 \usepackage{balance}
	 \usepackage{empheq}
	 \usepackage{mathtools}
	 \usepackage{algorithm}
\usepackage{algpseudocode}
\usepackage{tabularx}

\usepackage{hyperref}
\hypersetup{
    colorlinks=true,
    linkcolor=black,
    filecolor=black,
    urlcolor=black,
}

\mathchardef\Re="023C
\mathchardef\Im="023D

\newtheorem{theorem}{Theorem}
\newtheorem{lemma}{Lemma}
\newtheorem{proposition}{Proposition}
\newtheorem{corollary}{Corollary}

\newtheorem{definition}{Definition}
\newtheorem{proof}{Proof}

\newtheorem{remark}{Remark}

\newcommand{\norm}[1]{\left\lVert#1\right\rVert}

\date{\vspace{-5ex}}

\begin{document}
\title{
	First Order Methods For Globally Optimal Distributed Controllers\\ Beyond Quadratic Invariance
	 	}

	 \author{Luca Furieri and Maryam Kamgarpour\thanks{This research was gratefully funded by the European Union ERC Starting Grant CONENE. Luca Furieri and Maryam Kamgarpour are with the Automatic Control Laboratory, Department of Information Technology and Electrical Engineering, ETH Z\"{u}rich, Switzerland. E-mails: {\tt\footnotesize \{furieril, mkamgar\}@control.ee.ethz.ch}. The source codes of our simulations are available upon request. A preliminary version of this paper appeared on 21.08.2019 as a preprint at https://www.research-collection.ethz.ch/handle/20.500.11850/359817 with a different title.} 
}
	%

\maketitle
	\allowdisplaybreaks
	

\begin{abstract} 
We study the distributed Linear Quadratic Gaussian (LQG) control problem in discrete-time and finite-horizon, where the controller depends linearly on the history of the outputs and it is required to lie in a given subspace, e.g. to possess a certain sparsity pattern. It is well-known that this problem can be solved with convex programming within the Youla domain if and only if a condition known as Quadratic Invariance (QI) holds. In this paper, we first show that given QI sparsity constraints, one can directly descend the gradient of the cost function within the domain of output-feedback controllers and converge to a global optimum. Note that convergence is guaranteed despite non-convexity of the cost function. Second, we characterize a class of Uniquely Stationary (US) problems, for which first-order methods are guaranteed to converge to a global optimum. We show that the class of US problems is strictly larger than that of strongly QI problems and that it is not included in that of QI problems. We refer to Figure~\ref{fig:scheme} for details. Finally, we propose a tractable test for the US property. 
\end{abstract}

\section{Introduction}
\label{se:introduction}

The safe and efficient operation of emerging networked dynamical systems, such as the smart grid and autonomous vehicles, relies on the decision making of multiple interacting agents. Controlling these systems optimally is challenged by an inherent lack of information about the systems’ internal variables, possibly due to privacy concerns, geographic distance or the high cost of implementing a reliable communication network. The classical works of \cite{Witsenhausen, papadimitriou1986intractable} highlighted that, given information constraints, even simple instances of the Linear Quadratic Gaussian (LQG) control problem can result in highly intractable optimization tasks. 


A vast amount of literature has focused on approaching the distributed LQG problem and its variants with convex programming in the Youla parameter \cite{youla1976modern}. This enables utilizing efficient off-the-shelf software for numerical computation. A main challenge inherent to this approach is that the distributed control problem admits an exact convex reformulation if and only if the information constraints and the system dynamics interact in a Quadratically Invariant (QI) manner \cite{rotkowitz2006characterization,QIconvexity}. This limitation severely restricts the class of problems for which optimal distributed controllers can be computed in a tractable way. A variety of approximation methods and alternative controller implementations have henceforth been devised to  deal with the non-QI cases, based both on convex programming and nonlinear optimization. However, these approaches cannot compute a globally optimal sparse output-feedback controller in general. We refer the reader to  \cite{furieri2019sparsity,wang2019system,SDP,wang2018convex,dvijotham2015convex,lin2011augmented}
for a collection of recent results.

The recent years have witnessed a rapid growth of interest in developing learning-based, model-free techniques for optimal control problems. Specifically, some scenarios envision an unknown black-box system, for which an optimal behavior is obtained by observing the system's output trajectories in response to different  controllers and iteratively improving the control policy. In these cases, optimizing within the Youla domain is impractical because one is unable to recover the disturbance trajectories from the observed output trajectories for an unknown dynamical system. Therefore, model-free scenarios motivate optimizing  directly within the domain of output-feedback controllers, for instance, by devising gradient-descent based methods. Convergence of these methods to a global optimum was recently proven for the LQR problem in the non-distributed case \cite{fazel2018global,gravell2019sparse,gravell2019learning,bu2019lqr,Hesameddin2019global}. When carrying on  these methods to the distributed controller case, however, one can in general only guarantee convergence to a stationary point, which may not be a global optimum \cite{bu2019lqr,maartensson2009gradient,hassan2019data}. For the infinite-horizon and static-controller cases, this is mainly due to the set of stabilizing distributed controllers being disconnected in general \cite{feng2019exponential}. To the best of the authors' knowledge, classes of distributed control problems solvable to global optimality with first-order methods are yet to be characterized, and a connection with the QI notion is yet to be established. Furthermore, a condition that is more general than QI for global optimality certificates has not been found yet. Indeed, the QI notion is closely linked to using convex programming; this paper was driven by the intuition that less restrictive conditions for global optimality might exist by instead using first-order optimization methods directly in the domain of output-feedback controllers. We will show that this intuition indeed holds true.


Motivated as above, we investigate first-order methods for  the distributed LQG problem in discrete-time and finite-horizon. Our contributions are as follows. First, we show that given QI sparsity constraints, one can descend the gradient of the generally non-convex cost function in the output-feedback domain and always converge to a globally optimal distributed controller. We foresee that this method will enable devising learning-based policy gradient approaches for distributed control in future works. Second, we characterize a new class of  Uniquely Stationary (US) control problems, which can be solved to global optimality using first-order methods. We show that every strongly QI problem  is US and that there are instances of US problems which are neither strongly QI or QI. We refer to Figure~\ref{fig:scheme} for the details.

\emph{Paper structure: } Section~\ref{se:preliminaries} introduces the necessary notation and background. Section~\ref{sec:first-order} contains our  first result about global optimality given strong QI and a numerical example. Section~\ref{sec:US} establishes our results on  first-order methods for certificates of global optimality strictly beyond QI. We conclude the paper in Section~\ref{sec:conclusions}.

\section{Background and Problem Statement}
\label{se:preliminaries}
We start this section by providing the necessary notation. We then proceed with stating the distributed LQG problem and reviewing useful results about disturbance-feedback control policies and quadratic invariance.

\subsection{Notation}
We use $\mathbb{R}$ to denote the set of real numbers. The $(i,j)$-th element in a matrix $Y \in \mathbb{R}^{m \times n}$ is referred to as $Y_{i,j}$. We use $I_n$ to denote the identity matrix of size $n \times n$,  $0_{m \times n}$ to denote the zero matrix of size $m \times n$ . Whenever the subscripts are omitted, the dimensions are implied by the context. The symbols $\text{Im}(M)$ and $\text{Ker}(M)$ denote the range and the kernel of the linear operator associated with matrix $M$. We write $M=\text{blkdg}(M_1,\ldots,M_n)$ to denote a block-diagonal matrix where the blocks are the matrices $M_1,\ldots,M_n$. For a symmetric matrix $M=M^\mathsf{T}$ we write $M \succ 0$ (resp. $M \succeq 0$) if and only if it is positive definite (resp. positive semidefinite), that is its eigenvalues are strictly positive (resp. non-negative). For two matrices $M,P$ of any dimensions $M\otimes P$ denotes the Kronecker product and for two matrices of equal dimensions $M\odot P$ denotes the Hadamard product\footnote{$(M\odot P)_{i,j}=M_{i,j}P_{i,j}$}. For any matrix $K \in \mathbb{R}^{m \times n}$,  $\text{vec}(K) \in \mathbb{R}^{mn}$ is a vector obtained by stacking the columns of $K$ into a single column.  Given a binary matrix $X \in \{0,1\}^{m \times n}$, we define the associated \emph{sparsity subspace}  as
	\begin{align*}
\text{Sparse}(X)\hspace{-0.1cm}:=\{&Y\hspace{-0.1cm} \mid  Y_{i,j}\hspace{-0.1cm}=0 ~~\text{for all } i,j \text{ such that  } X_{i,j}=0 \;  \}\,.
	\end{align*}
 Similarly, given $Y \in \mathbb{R}^{m \times n}$, we define $X = \text{Struct}(Y)$ as the binary matrix such that $X_{i,j}=0 $ if $Y_{i,j}=0$ and $X_{i,j}=1$ otherwise. Let $X, \hat{X} \in \{0,1\}^{m \times n}$ and $Z \in \{0,1\}^{n \times p}$ be binary matrices. We adopt the following conventions: $X + \hat{X} := \text{Struct}(X + \hat{X})$, $ XZ:=\text{Struct}(XZ)$, $X \leq \hat{X}$ if and only if $X_{i,j}\leq \hat{X}_{i,j}\;\forall i,j$. The Euclidean norm of a vector $v \in \mathbb{R}^n$ is denoted by $\norm{v}_2^2=v^\mathsf{T}v$  and the Frobenius norm of a matrix $M \in \mathbb{R}^{m \times n}$ is denoted by $\norm{M}_{F}^2=\text{Trace}(M^\mathsf{T}M)$. Given a matrix $K \in \mathbb{R}^{m \times n}$ and a continuously differentiable function $J:\mathbb{R}^{m \times n} \rightarrow \mathbb{R}$ we define $\nabla J(K)$ as the $m \times n$ matrix such that $\nabla J(K)_{i,j}=\frac{\partial J(K)}{\partial K_{i,j}}$. For a vector $v \in \mathbb{R}^n$ and a function $f:\mathbb{R}^n \rightarrow \mathbb{R}$ we denote the gradient by  $\nabla f(v) \in \mathbb{R}^n$ and the Hessian by $\nabla^2 f(v) \in \mathbb{R}^{n \times n}$. Given a subspace $\mathcal{K} \subseteq \mathbb{R}^{m \times n}$ we denote its orthogonal complement as $\mathcal{K}^\perp$. The symbol $\mathcal{N}(\mu,\Sigma)$ denotes the normal distribution with expected value $\mu \in \mathbb{R}^n$ and covariance matrix $\Sigma\in \mathbb{R}^{n \times n} \succeq 0$, and $x\sim \mathcal{N}(\mu,\Sigma)$ indicates that  $x\in \mathbb{R}^n$ follows the distribution $\mathcal{N}(\mu,\Sigma)$. For a subspace $\mathcal{K} \subseteq \mathbb{R}^n$, $\Pi_\mathcal{K}(\cdot)$ denotes the projection operator on $\mathcal{K}$.

	\subsection{Problem Setup}\label{se:problems}

	
	
	We consider time-varying linear systems in discrete-time
	\begin{align}
	\label{eq:sys_disc}
	x_{t+1}&=A_tx_t+B_tu_t+w_t\,,\\
	y_t&=C_tx_t+v_t \nonumber\,,
	\end{align}
	where $x_t \in \mathbb{R}^n$ is the system state at time $t$ affected by additive noise $w_t \sim\mathcal{N}(0,\Sigma^w_t)$ with $x_0 \sim \mathcal{N}(\mu_0,\Sigma_0)$ ,  $y_t\in \mathbb{R}^p$ is the output at time $t$ affected by  additive noise $v_t\sim\mathcal{N}(0,\Sigma_t^v)$ and  $u_t \in \mathbb{R}^m$ is the control input at time $t$. We assume that $\Sigma_0,\Sigma_t^w\succeq 0$ and $\Sigma_t^v\succ 0$ for all $t$. We consider the evolution of (\ref{eq:sys_disc}) in finite-horizon for $t=0,\ldots N$, where $N \in \mathbb{N}$. By defining the matrices $\mathbf{A}=\text{blkdg}(A_0,\ldots,A_{N})$,
\begin{equation*}
\mathbf{B}=\begin{bmatrix}
\text{blkdg}(B_0,\ldots,B_{N\text{-}1})\\
0_{n \times mN}
\end{bmatrix}\,,\mathbf{C}=\begin{bmatrix}
\text{blkdg}(C_0,\ldots,C_{N\text{-}1})^\mathsf{T}\\0_{n\times pN}
\end{bmatrix}^\mathsf{T}\,,
\end{equation*}
 and  the vectors $\mathbf{x}=\begin{bmatrix}x_0^\mathsf{T}&\ldots&x_N^\mathsf{T} \end{bmatrix}^\mathsf{T}\in \mathbb{R}^{n(N+1)}$, $\mathbf{y}=\begin{bmatrix}y_0^\mathsf{T}&\ldots&y_{N-1}^\mathsf{T} \end{bmatrix}^\mathsf{T}\in \mathbb{R}^{pN}$, $\mathbf{u}=\begin{bmatrix}u_0^\mathsf{T}&\ldots&u_{N-1}^\mathsf{T}\end{bmatrix}^\mathsf{T}$ $\in \mathbb{R}^{mN}$, $\mathbf{w}=\begin{bmatrix}x_0^\mathsf{T}&w_0^\mathsf{T}&\ldots&w_{N-1}^\mathsf{T}\end{bmatrix}^\mathsf{T}\in \mathbb{R}^{n(N+1)}$ and $\mathbf{v}=\begin{bmatrix}v_0^\mathsf{T}&\ldots&v_{N-1}^\mathsf{T}\end{bmatrix}^\mathsf{T}\in \mathbb{R}^{pN}$, and the shift matrix
\begin{equation*}
\mathbf{Z}=\begin{bmatrix}
0_{n \times nN}&0_{n \times n}\\
I_{nN}&0_{nN \times n}
\end{bmatrix}\,,
\end{equation*}
we can  write the system (\ref{eq:sys_disc}) compactly as
\begin{align}
\label{eq:system_compact}
&\mathbf{x}=\mathbf{P}_{11}\mathbf{w}+\mathbf{P}_{12}\mathbf{u}\,, \quad \mathbf{y}=\mathbf{Cx}+\mathbf{v}\,,
\end{align}
where  $\mathbf{P}_{11}=(I-\mathbf{Z}\mathbf{A})^{-1}$ and $\mathbf{P}_{12}=(I-\mathbf{Z}\mathbf{A})^{-1}\mathbf{Z}\mathbf{B}$. In this paper we consider output-feedback policies of the form
\begin{equation}
\label{eq:control_input_def}
\mathbf{u}=\mathbf{Ky}, \quad \mathbf{K} \in \mathcal{K}\,,
\end{equation}
where $\mathcal{K}$ is a subspace that $1)$ ensures causality of the feedback policy by  forcing  to $0$ those entries of $\mathbf{K}$ corresponding to future outputs, $2)$ may encode arbitrary time-varying spatio-temporal sparsity constraints for distributed control as per \cite{furieri2019unified}, and $3)$ can impose that the control policy is memory-less and time-independent in the sense that $\mathbf{K}=I_N \otimes K$ for some $K \in \mathbb{R}^{m \times p}$. 
%


Our goal is to compute $\mathbf{K} \in \mathcal{K}$  that minimizes  the expected value of a quadratic cost in the states and the inputs:
\begin{equation}
\label{eq:cost}
J(\mathbf{K})\hspace{-0.06cm}:=\hspace{-0.06cm}\mathbb{E}_{\mathbf{w},\mathbf{v}}\hspace{-0.06cm}\left[\sum_{t=0}^{N-1}\hspace{-0.1cm}\left(x_t^\mathsf{T}M_tx_t\hspace{-0.06cm}+\hspace{-0.06cm}u_t^\mathsf{T}R_tu_t\right)\hspace{-0.06cm}+\hspace{-0.06cm}x_N^\mathsf{T}M_Nx_N\right]\hspace{-0.15cm}\,,
\end{equation}	
where  $M_t \succeq 0$ and $R_t \succ 0$ for every $t$. 
\begin{remark}
\emph{The problem of minimizing (\ref{eq:cost}) is known as the Linear Quadratic Gaussian (LQG) problem. It is well-known that a time-invariant and memory-less control policy (commonly denoted as \emph{static})  of the form $u_t=Ky_t$ achieves global optimality when $N \rightarrow \infty$ and there are no subspace constraints to comply with. For the finite-horizon and/or constrained cases,  a time-varying control policy with memory (commonly denoted as \emph{dynamic}) achieves higher performance in general. In this paper, we therefore consider dynamic linear policies as in (\ref{eq:control_input_def}).} 
\end{remark}

From (\ref{eq:system_compact})-(\ref{eq:control_input_def}) we derive the closed-loop equations:
\begin{align}
&\mathbf{x}=(I-\mathbf{P}_{12}\mathbf{KC})^{-1}(\mathbf{P}_{11}\mathbf{w}+\mathbf{P}_{12}\mathbf{Kv})\,, \nonumber\\
&\mathbf{y}=\mathbf{C}(I-\mathbf{P}_{12}\mathbf{KC})^{-1}\mathbf{P}_{11}\mathbf{w}+(I-\mathbf{CP}_{12}\mathbf{K})^{-1}\mathbf{v}\,,\label{eq:closed_loop_K}\\
&\mathbf{u}=\mathbf{KC}(I-\mathbf{P}_{12}\mathbf{KC})^{-1}\mathbf{P}_{11}\mathbf{w}+\mathbf{K}(I-\mathbf{CP}_{12}\mathbf{K})^{-1}\mathbf{v}\,.\nonumber
\end{align}
By defining $\mathbf{M}=\text{blkdg}(M_0,M_1,\ldots,M_N)$, $\mathbf{R}=\text{blkdg}(R_0,\ldots R_{N-1})$, $\mathbf{\Sigma}_w=\text{blkdg}(\Sigma_0,\Sigma_0^w,\ldots, \Sigma_{N-1}^w)$,  $\mathbf{\Sigma}_v=\text{blkdg}(\Sigma_0^v,\ldots, \Sigma_{N-1}^v)$, $\bm{\mu}_w=\begin{bmatrix}\mu_0^\mathsf{T}&0&\ldots&0\end{bmatrix}^\mathsf{T}$ the cost function (\ref{eq:cost}) can thus be written as
\begin{align}
J(\mathbf{K})&=\norm{\mathbf{M}^{\frac{1}{2}}(I-\mathbf{P}_{12}\mathbf{KC})^{-1}\mathbf{P}_{11}\mathbf{\Sigma}_w^{\frac{1}{2}}}_F^2+\norm{\mathbf{M}^{\frac{1}{2}}\mathbf{P}_{12}\mathbf{K}(I-\mathbf{C}\mathbf{P}_{12}\mathbf{K})^{-1}\mathbf{\Sigma}_v^{\frac{1}{2}}}_F^2\nonumber\\
&+\norm{\mathbf{R}^{\frac{1}{2}}\mathbf{K}(I-\mathbf{C}\mathbf{P}_{12}\mathbf{K})^{-1}\mathbf{C}\mathbf{P}_{11}\mathbf{\Sigma}_w^{\frac{1}{2}}}_F^2+\norm{\mathbf{R}^{\frac{1}{2}}\mathbf{K}(I-\mathbf{C}\mathbf{P}_{12}\mathbf{K})^{-1}\mathbf{\Sigma}_v^{\frac{1}{2}}}_F^2\label{eq:cost_K}\\
&+\norm{\mathbf{M}^{\frac{1}{2}}(I-\mathbf{P}_{12}\mathbf{KC})^{-1}\mathbf{P}_{11}\bm{\mu}_w}_2^2+\norm{\mathbf{R}^{\frac{1}{2}}\mathbf{K}(I-\mathbf{C}\mathbf{P}_{12}\mathbf{K})^{-1}\mathbf{C}\mathbf{P}_{11}\bm{\mu}_w}_2^2\,.\nonumber
\end{align}
A derivation of $J(\mathbf{K})$ as per (\ref{eq:cost_K}) is reported in the Appendix.
\begin{remark}
\label{re:polynomial}
\emph{Note that $J(\mathbf{K})$ is a multivariate polynomial in the entries of $\mathbf{K}$. Indeed, one can verify
\begin{equation*}
(I-\mathbf{CP}_{12}\mathbf{K})^{-1}=\sum_{i=0}^{N}(\mathbf{CP}_{12}\mathbf{K})^i\,,
\end{equation*}
due to the fact that each $p \times p$ block on the diagonal of $\mathbf{CP}_{12}\mathbf{K}$ is the zero matrix by construction, and hence $(\mathbf{CP}_{12}\mathbf{K})^i=0_{pN \times pN}$ for every  $i\geq N+1$. }
\end{remark}

 To summarize, in this paper we are interested in solving the following optimization problem $\mathcal{P}_K$:
 \begin{alignat*}{3}
 &\textbf{Problem} ~~&&\mathcal{P}_K\\
 &\min_{\mathbf{K} \in \mathcal{K}}&&J(\mathbf{K})\,,
 \end{alignat*}
which might be non-convex due to $J$ being non-convex in $\mathbf{K}$ in general.
\subsection{Disturbance-feedback strategies}
The classical way to deal with the non-convexity of $J(\mathbf{K})$ is to parametrize the output-feedback policy $\mathbf{u}=\mathbf{Ky}$ in terms of an equivalent disturbance-feedback policy $\mathbf{u}=\mathbf{QCP}_{11}\mathbf{w}+\mathbf{Qv}$  \cite{colombino2017mutually,furieri2019unified}. Such parametrization is akin to the \emph{Youla parametrization} \cite{youla1976modern}. Similarly to \cite{colombino2017mutually,furieri2019unified}, we have the following result, whose proof is reported in the Appendix.


\begin{lemma}
\label{le:strictly_convex}
Let us define function $\tilde{J}:\mathbb{R}^{mN \times pN}\rightarrow \mathbb{R}$ as
\begin{align}
\tilde{J}(\mathbf{Q})&=\norm{\mathbf{M}^{\frac{1}{2}}(I+\mathbf{P}_{12}\mathbf{QC})\mathbf{P}_{11}\mathbf{\Sigma}_w^{\frac{1}{2}}}_F^2+\norm{\mathbf{M}^{\frac{1}{2}}\mathbf{P}_{12}\mathbf{Q}\mathbf{\Sigma}_v^{\frac{1}{2}}}_F^2+\norm{\mathbf{R}^{\frac{1}{2}}\mathbf{Q}\mathbf{C}\mathbf{P}_{11}\mathbf{\Sigma}_w^{\frac{1}{2}}}_F^2\label{eq:cost_Q}\\
&+\norm{\mathbf{R}^{\frac{1}{2}}\mathbf{Q}\mathbf{\Sigma}_v^{\frac{1}{2}}}_F^2+\norm{\mathbf{R}^{\frac{1}{2}}\mathbf{Q}\mathbf{C}\mathbf{P}_{11}\bm{\mu}_w}_2^2+\norm{\mathbf{M}^{\frac{1}{2}}(I+\mathbf{P}_{12}\mathbf{QC})\mathbf{P}_{11}\bm{\mu}_w}_2^2\nonumber\,.
\end{align}
Let  $h:\mathbb{R}^{mN\times pN} \rightarrow \mathbb{R}^{mN\times pN}$ be the bijection defined as 
\begin{align*}
&h(\mathbf{Q},\mathbf{CP}_{12})=(I+\mathbf{Q}\mathbf{CP}_{12})^{-1}\mathbf{Q}\,.
\end{align*}The following facts hold.
\begin{enumerate}
\item $\tilde{J}(\mathbf{Q})$ is strictly convex and quadratic in $\mathbf{Q}$.
\item $\tilde{J}(h^{-1}(\mathbf{K},\mathbf{CP}_{12}))=J(\mathbf{K})$ for all $\mathbf{K} \in \mathbb{R}^{mN \times pN}$.
\item $\tilde{J}(\mathbf{Q})=J(h(\mathbf{Q},\mathbf{CP}_{12}))$ for all $\mathbf{Q} \in \mathbb{R}^{mN \times pN}$.
\end{enumerate}
\end{lemma}

\vspace{0.2cm}

In other words, the nonlinear change of coordinates induced by $h$ allows expressing the non-convex cost function $J(\mathbf{K})$ in (\ref{eq:cost_K}) as the convex function $\tilde{J}(\mathbf{Q})$ in (\ref{eq:cost_Q}).  
Last, we characterize the following property of $J(\mathbf{K})$ to be exploited in Section~\ref{sec:first-order} and Section~\ref{sec:US}.  The corresponding proof is reported in the Appendix.

%

\begin{lemma}
\label{le:bounded_level_sets}
Let $\mathbf{K}_0\in \mathbb{R}^{mN \times pN}$ and define the sublevel set of $J(\mathbf{K}_0)$ as $\mathcal{L}:=\{\mathbf{K}|~J(\mathbf{K}) \leq J(\mathbf{K}_0)\}$. The sublevel set $\mathcal{L}$ is bounded for any $\mathbf{K}_0$.
\end{lemma}


	\subsection{Quadratic invariance}
 Since $\tilde{J}$ is convex and it corresponds to $J$ up to a nonlinear change of coordinates, one may exploit $\tilde{J}$ for convex computation of constrained controllers. In particular, if and only if a property denoted as Quadratic Invariance (QI) holds \cite{rotkowitz2006characterization,QIconvexity}, one can solve a convex program in $\mathbf{Q}$ that is equivalent to $\mathcal{P}_K$. For our finite-horizon setting, it is convenient to review the notions of QI  and strong QI and recall the corresponding convexity result from \cite{furieri2019unified}.
 \begin{definition}
\emph{ A subspace $\mathcal{K} \subset \mathbb{R}^{mN\times pN}$ is \emph{QI} with respect to $\mathbf{CP}_{12}$ if and only if 
  \begin{equation*}
 \mathbf{K}\mathbf{CP}_{12}\mathbf{K} \in \mathcal{K}\,, \quad \forall \mathbf{K}\in \mathcal{K}\,.
 \end{equation*}
 and it is \emph{strongly QI} with respect to $\mathbf{CP}_{12}$ if and only if
 \begin{equation*}
 \mathbf{K}_1\mathbf{CP}_{12}\mathbf{K}_2 \in \mathcal{K}\,, \quad \forall \mathbf{K}_1,\mathbf{K}_2 \in \mathcal{K}\,.
 \end{equation*}}
 \end{definition}
 
Note that a general subspace is QI if it is strongly QI, but not vice-versa; instead, a sparsity subspace $\text{Sparse}(\mathbf{S})$ is QI \emph{if and only if} it is strongly QI \cite{rotkowitz2006characterization}. Now, notice that by Lemma~\ref{le:strictly_convex} our original  problem $\mathcal{P}_K$ is equivalent to 
 
\begin{equation}
\label{eq:nonconvex_Q}
\min_{\mathbf{Q} \in h^{-1}(\mathcal{K},\mathbf{CP}_{12})} \tilde{J}(\mathbf{Q})\,.
\end{equation}
  The  QI result in finite-horizon is that problem (\ref{eq:nonconvex_Q}) is convex if and only QI holds. We refer to \cite{furieri2019unified,rotkowitz2006characterization,QIconvexity} for details.
\begin{theorem}[QI]
 \label{th:QI}
The following three statements are equivalent.
\begin{enumerate}
\item The set $h^{-1}(\mathcal{K},\mathbf{CP}_{12})=-h(\mathcal{K},\mathbf{CP}_{12})$ is convex.
\item $\mathcal{K}$ is \emph{QI} with respect to $\mathbf{CP}_{12}$.
\item $h^{-1}(\mathcal{K},\mathbf{CP}_{12})=\mathcal{K}$.
\end{enumerate}
 \end{theorem}
 
 It follows from Theorem~\ref{th:QI} that problem $\mathcal{P}_K$ is equivalent to a convex program, and in particular equivalent to
 \begin{equation}
\label{eq:convex_Q}
\min_{\mathbf{Q} \in \mathcal{K}} \tilde{J}(\mathbf{Q})\,,
\end{equation}
 if and only if QI holds.
 
 \vspace{0.1cm}


As we have observed in Section~\ref{se:introduction}, if the system model was unknown and we only had black-box simulation access to the cost function, we would not be able to optimize within the $\mathbf{Q}$ domain due to the mapping $h$ being unknown. Moreover, it would be highly desirable to step beyond the long-standing QI limitation, which is inherent to using convex programming in the $\mathbf{Q}$ domain.  Motivated as above, the rest of the paper develops a first-order gradient-descent method to solve $\mathcal{P}_K$ to global optimality directly in the $\mathbf{K}$ domain.  

 \section{First-order method for globally optimal sparse controllers\\ given QI}
 \label{sec:first-order}

 
  In this section we focus our attention on sparsity subspace constraints for the synthesis of distributed controllers complying with arbitrary information structures \cite{furieri2019unified}. For a sparsity constraint $\mathbf{K}\in \text{Sparse}(\mathbf{S})$, the set of stationary points for problem $\mathcal{P}_K$ is defined as follows:
 
  \begin{definition}
 \emph{Consider problem $\mathcal{P}_K$ with $\mathcal{K}= \text{Sparse}(\mathbf{S})$. A controller $\overline{\mathbf{K}} \in \text{Sparse}(\mathbf{S})$ is a stationary point for $\mathcal{P}_K$ if and only if 
 \begin{equation}
 \label{eq:structured_stationry}
 \nabla J(\overline{\mathbf{K}}) \in \text{Sparse}(\mathbf{S})^\perp = \text{Sparse}(\mathbf{S}^c)\,,
 \end{equation}
 where $\mathbf{S}^c$ is the binary matrix that has a $0$ wherever $\mathbf{S}$ has a $1$, and a $1$ wherever $\mathbf{S}$ has a $0$. }
 \end{definition}

 In general, a stationary point as in (\ref{eq:structured_stationry}) could be a local minimum, a local maximum or a saddle point for  $\mathcal{P}_K$. In the next lemma, we show that the set of stationary points for $\mathcal{P}_K$ corresponds to that of stationary points for problem (\ref{eq:convex_Q}) when strong QI holds. The proof is mainly based on \cite[Lemma~1]{colombino2017mutually}. We report it in the Appendix for completeness.
 
 \begin{lemma}
 \label{le:equivalence_stationary}
  Suppose that the subspace $\mathcal{K}$ is strongly \emph{QI} with respect to $\mathbf{CP}_{21}$, and let $\overline{\mathbf{K}} \in \mathcal{K}$. Also define $\overline{\mathbf{Q}}=h^{-1}(\overline{\mathbf{K}},\mathbf{CP}_{12})$. We have that
  \begin{equation*}
 \nabla \tilde{J}(\mathbf{\overline{\mathbf{Q}}}) \in \mathcal{K}^\perp \iff  \nabla J(\overline{\mathbf{K}}) \in \mathcal{K}^\perp\,.
  \end{equation*}
 \end{lemma}
 
 Notice that since any QI sparsity subspace is also strongly QI \cite{rotkowitz2006characterization}, Lemma~\ref{le:equivalence_stationary} holds for all the arbitrary QI information structures characterized in \cite{furieri2019unified}. 
 \subsection{Global optimality of gradient-descent}
By exploiting Lemma~\ref{le:equivalence_stationary} our first result establishes that if $\text{Sparse}(\mathbf{S})$ is QI with respect to $\mathbf{CP}_{12}$, any stationary point of $\mathcal{P}_K$ is a global optimum.
\begin{theorem}
\label{th:QI_uniquely_stationary}
Suppose that $\text{\emph{Sparse}}(\mathbf{S})$ is \emph{QI} with respect to $\mathbf{CP}_{12}$ and let $\mathbf{K}^\star \in \text{\emph{Sparse}}(\mathbf{S})$ be a stationary point of $J(\mathbf{K})$. Then, 
\begin{align*}
\mathbf{K}^\star \in \arg \min_{\mathbf{K} \in \text{\emph{Sparse}}(\mathbf{S})} J(\mathbf{K})\,.
\end{align*}
\end{theorem}
\begin{proof}
\emph{By Theorem~\ref{th:QI}, $\mathcal{P}_K$ is equivalent to (\ref{eq:convex_Q}). Since problem (\ref{eq:convex_Q}) is convex, every $\mathbf{Q}^\star \in \text{Sparse}(\mathbf{S})$ such that $\nabla \tilde{J}(\mathbf{Q}^\star) \in \text{Sparse}(\mathbf{S}^c)$ (that is, $\mathbf{Q}^\star$ is a stationary point) is a global optimum and thus achieves the optimal cost $J^\star$. Let $\mathbf{K}^\star=h(\mathbf{Q}^\star,\mathbf{CP}_{12})$. Now remember that $\text{Sparse}(\mathbf{S})$ is QI if and only if it is strongly QI \cite{rotkowitz2006characterization}. By Lemma~\ref{le:equivalence_stationary}  $\nabla J(\mathbf{K}^\star) \in \text{Sparse}(\mathbf{S}^c)$,  and hence $\mathbf{K}^\star$ is a stationary point for $J(\mathbf{K})$. Since $\tilde{J}(\mathbf{Q})=J(h(\mathbf{Q},\mathbf{CP}_{12}))$ for every $\mathbf{Q}$ by definition, we have that  $J(\mathbf{K}^\star)=\tilde{J}(\mathbf{Q}^\star)=J^\star$ and thus $\mathbf{K}^\star$ is optimal.  By Lemma~\ref{le:equivalence_stationary}, there can be no other stationary point $\overline{\mathbf{K}} \in \text{Sparse}(\mathbf{S}^c)$ such that $J(\overline{\mathbf{K}})=\overline{J}>J^\star$; otherwise, $\overline{\mathbf{Q}}=h^{-1}(\overline{\mathbf{K}},\mathbf{PC}_{12})$ would also be a stationary point for problem (\ref{eq:convex_Q}) with cost $\overline{J}>J^\star$, which is a contradiction due to (\ref{eq:convex_Q}) being convex.}
\end{proof}

\begin{remark}
\label{re:generalization}
\emph{Theorem~\ref{th:QI_uniquely_stationary} trivially generalizes to any subspace constraint $\mathcal{K}$ that is  strongly QI, as the key Lemma~\ref{le:equivalence_stationary} holds for any strongly QI subspace. In Theorem~\ref{th:QI_uniquely_stationary}, we decided to specialize the result to the most common case of sparsity constraints in the interest of clarity.} 
\end{remark}



Theorem~\ref{th:QI_uniquely_stationary} leads to a fundamental insight: under QI sparsity constraints, if we can find any stationary point of the generally non-convex function $J(\mathbf{K})$, this point is certified to be a globally optimal solution to $\mathcal{P}_K$. Based on this observation, we develop a gradient-descent method that solves $\mathcal{P}_K$ to global optimality for QI sparsity constraints. 
\begin{theorem}
\label{th:GD}
Suppose $\text{\emph{Sparse}}(\mathbf{S})$ is \emph{QI} with respect to $\mathbf{CP}_{12}$. Let $\mathbf{K}_0 \in \text{\emph{Sparse}}(\mathbf{S})$ be an initial output-feedback control policy, and consider the iteration
\begin{align}
&\mathbf{K}_{t+1}=\mathbf{K}_t-\eta_t \nabla J(\mathbf{K}_t) \odot \mathbf{S}\,.\label{eq:iteration}
\end{align}
Then, $\mathbf{K}_t \in \text{\emph{Sparse}}(\mathbf{S})$ for every $t$ and there exists  $\eta_t$ for every $t$ such that
\begin{equation*}
\lim_{t \rightarrow \infty}J(\mathbf{K}_t)=J^\star\,,
\end{equation*}
where $J^\star$ is the optimal value of problem $\mathcal{P}_K$.
\end{theorem}

The proof of Theorem~\ref{th:GD} uses Lemma~\ref{le:bounded_level_sets} and the following four Lemmas. The proofs of Lemmas~\ref{le:Zoutendijk}, \ref{le:existence_Wolfe} and \ref{le:Lipschitzianity} can be found in \cite[Theorem~3.2]{nocedal2006numerical}, \cite[Lemma~3.1]{nocedal2006numerical} and \cite[Proposition~5.7]{aragon2019nonlinear} respectively. We prove Lemma~\ref{le:unconstrained} in the Appendix.
\begin{lemma}
\label{le:Zoutendijk}
Let $f:\mathbb{R}^n \rightarrow \mathbb{R}$ be bounded below and consider the iteration
\begin{equation}
\label{eq:iteration_zoutendijk}
x_{t+1}=x_t-\eta_t \nabla f(x_t)\,,
\end{equation}
where $\eta_t$ satisfies the Wolfe conditions:
\begin{equation}
f(x_t-\eta_t\nabla f(x_t))\leq f(x_t)-c_1 \eta_t\|\nabla f(x_t)\|^2_2\,, \label{eq:Wolfe_1}
\end{equation}
\begin{equation}
-\nabla f(x_t-\eta_t \nabla f(x_t))^\mathsf{T} \nabla f(x_t) \geq -c_2 \|\nabla f(x_t)\|_2^2\,,\label{eq:Wolfe_2}
\end{equation}
 for some $0<c_1<c_2<1$ and every $t$. Let $f$ be continuously differentiable in an open set $\mathcal{U}$ containing the sublevel set $\mathcal{L}=\{x:f(x)\leq f(x_0)\}$, where $x_0$ is the starting point of the iteration (\ref{eq:iteration_zoutendijk}). Assume that $\nabla f$ is Lipschitz continuous on $\mathcal{U}$.   Then, $\lim_{t\rightarrow \infty}\nabla f(x_t) =0\,.
$

\end{lemma}

\begin{lemma}
\label{le:existence_Wolfe}
Suppose $f: \mathbb{R}^n \rightarrow \mathbb{R}$ is continuously differentiable and bounded below. Then, for any $0<c_1<c_2<1$, there exist intervals of step lengths satisfying the Wolfe conditions (\ref{eq:Wolfe_1})-(\ref{eq:Wolfe_2}).
\end{lemma}
\begin{lemma}
\label{le:Lipschitzianity}
Let $f\hspace{-0.1cm}:\hspace{-0.05cm}\mathbb{R}^n \hspace{-0.1cm}\rightarrow\hspace{-0.1cm} \mathbb{R}$ be twice continuously differentiable on an open convex set $\mathcal{U} \subseteq \mathbb{R}^n$, and suppose that  $\nabla^2f$ is bounded on $\mathcal{U}$. Then, $\nabla f$ is Lipschitz continuous on $\mathcal{U}$.
\end{lemma}
\begin{lemma}
\label{le:unconstrained}
Let $f:\mathbb{R}^n\rightarrow \mathbb{R}$ be a continuously differentiable function and $\mathcal{K} \subseteq \mathbb{R}^n$ be a subspace. Let $\underline{f}:\mathbb{R}^r \rightarrow \mathbb{R}$ be defined as $\underline{f}(\bm{\alpha})=f(M\bm{\alpha})$, where $\text{\emph{Im}}(M)=\mathcal{K}$ and $r$ is the dimension of $\mathcal{K}$. Then:
\begin{enumerate}
\item $\min_{\mathbf{x} \in \mathcal{K}} f(\mathbf{x})=\min_{\bm{\alpha} \in \mathbb{R}^r}\underline{f}(\bm{\alpha})$.
\item $\nabla \underline{f}(\bm{\alpha})=0 \iff \nabla f(M \bm{\alpha}) \in \mathcal{K}^\perp$.
\end{enumerate}

\end{lemma}

\vspace{6pt}

We are now ready to prove Theorem~\ref{th:GD}.

\begin{proof}[Theorem~\ref{th:GD}]
\emph{Denote $\mathbf{k}=\text{vec}(\mathbf{K})$ and let $f:\mathbb{R}^{mpN^2 } \rightarrow \mathbb{R}$ be the function such that $f(\mathbf{k})=J(\mathbf{K})$ for every $\mathbf{K}\in \mathbb{R}^{mN\times pN}$. Clearly, if $\mathbf{k}_0=\text{vec}(\mathbf{K}_0)$ the iterations of (\ref{eq:iteration}) are equivalent to those of
\begin{align}
\label{eq:iteration_kvec}
\mathbf{k}_{t+1}&=\mathbf{k}_t-\eta_t \nabla f(\mathbf{k}_t)\odot \text{vec}(\mathbf{S})=\mathbf{k}_t-\eta_t \Pi_{\text{Sparse}(\text{vec}(\mathbf{S}))}\left(\nabla f(\mathbf{k}_t)\right)\,.
\end{align}
Now let $M$ be such that its columns are an orthonormal basis of $\text{Sparse}(\text{vec}(\mathbf{S}))$. Consider the iteration 
\begin{equation}
\label{eq:iteration_alpha}
\bm{\alpha}_{t+1}=\bm{\alpha}_t-\eta_t \nabla \underline{f}(\bm{\alpha}_t)\,,
\end{equation}
where $\underline{f}$ is such that $\underline{f}(\bm{\alpha})=f(M\bm{\alpha})$ for every $\bm{\alpha}$. Let $\mathbf{k}_0=M\bm{\alpha}_0$ and suppose that $\mathbf{k}_t=M \bm{\alpha}_t$. Then, by (\ref{eq:iteration_kvec}), (\ref{eq:iteration_alpha}) and noting that $M^\mathsf{T}M=I$ and $\nabla \underline{f}(\mathbf{\bm{\alpha}}_t)=M^\mathsf{T}\nabla f(M\bm{\alpha}_t)$:
\begin{align*}
\mathbf{k}_{t+1}&=M\bm{\alpha}_t-\eta_t M(M^\mathsf{T}M)^{-1}M^\mathsf{T}\nabla f(M\bm{\alpha}_t)=M\bm{\alpha}_{t+1}\,.
\end{align*}
 We conclude by induction that  $\text{vec}(\mathbf{K}_t)=M\bm{\alpha}_t$ for every $t$. 
Let us choose $\eta_t$ satisfying (\ref{eq:Wolfe_1})-(\ref{eq:Wolfe_2}). Notice that, according to Lemma~\ref{le:existence_Wolfe}, a choice for $\eta_t$ exists for every $t$ because $f$ is continuously differentiable and bounded below by $0$. By Lemma~\ref{le:bounded_level_sets}  we obtain that the sublevel set $\mathcal{L}=\{\bm{\alpha}|~\underline{f}(\bm{\alpha}) \leq f(\bm{\alpha}_0)\}$ is bounded. Consider an open, convex and bounded set $\mathcal{U}$ that contains $\mathcal{L}$. Since $\underline{f}(\bm{\alpha})$ is a multivariate polynomial, every entry of its Hessian matrix $\nabla^2  \underline{f}(\bm{\alpha})$ is also a multivariate polynomial and is thus bounded on $\mathcal{U}$. By Lemma~\ref{le:Lipschitzianity}, we deduce that $\nabla \underline{f}$ is Lipschitz continuous. By Lemma~\ref{le:Zoutendijk},
\begin{equation}
\label{eq:convergence}
\lim_{t\rightarrow \infty}\nabla \underline{f}(\bm{\alpha}_t)=0\,.
\end{equation}}

\emph{Let $\bm{\alpha}^\star=\lim_{t\rightarrow \infty}\bm{\alpha}_t$ and $\mathbf{K}^\star \in \text{Sparse}(\mathbf{S})$ the corresponding output-feedback controller according to $\text{vec}(\mathbf{K}^\star)=M\bm{\alpha}^\star$. Since $\nabla \underline{f}(\bm{\alpha}^\star)=0$, by Lemma~\ref{le:unconstrained} $\nabla J(\mathbf{K}^\star) \in \text{Sparse}(\mathbf{S}^c)$ and hence $\nabla J(\mathbf{K}^\star) \odot \mathbf{S}=0$, that is, iteration (\ref{eq:iteration}) converges to a stationary point of $\mathcal{P}_K$.  By Theorem~\ref{th:QI_uniquely_stationary}, $\mathbf{K}^\star$ is a globally optimal solution for $\mathcal{P}_K$ because $\text{Sparse}(\mathbf{S})$ is QI with respect to $\mathbf{CP}_{12}$.}
\end{proof}
A choice for $\eta_t$ satisfying the Wolfe conditions (\ref{eq:Wolfe_1})-(\ref{eq:Wolfe_2}) can be found by using, for instance, the bisection algorithm reported in \cite[Proposition~5.5]{aragon2019nonlinear}, which always converges in a finite number of iterations. We conclude this section by providing a numerical example.

\subsection{Numerical example}
Motivated by the example system of \cite{rotkowitz2006characterization}, we consider system (\ref{eq:sys_disc}) and the cost function (\ref{eq:cost_K}) with
\begin{align*}
&A_t={\scriptsize\begin{bmatrix}
1.6&0&0&0&0\\
0.5&1.6&0&0&0\\
2.5&2.5&-1.4&0&0\\
-2&1&-2&0.1&0\\
0&2&0&-0.5&1.1
\end{bmatrix}}\,,
\end{align*}
and $B_t=I,$ $C_t=I$, $M_t=I$, $R_t=I$, $\mathbf{\Sigma}_{w}=I$, $\mathbf{\Sigma}_v=I$,   $\mu_0=\begin{bmatrix}1&-1&2&-3&3\end{bmatrix}^\mathsf{T}$. We set a horizon of $N=3$. Our goal is to compute a controller $\mathbf{K}$ with a given sparsity that minimizes the cost (\ref{eq:cost_K}). Specifically, we aim to solve $\mathcal{P}_K$ with $\mathcal{K}=\text{Sparse}(\mathbf{S})$ and $\mathbf{S}=T\otimes S$, 
where $T_{i,j}=1$ if $j\leq i$ and $T_{i,j}=0$ otherwise, and
\begin{equation*}
S={\scriptsize \begin{bmatrix}
0&0&0&0&0\\
0&1&0&0&0\\
0&1&0&0&0\\
0&1&0&0&0\\
0&1&0&0&1
\end{bmatrix}}\,.
\end{equation*}
The total number of scalar decision variables is $|\mathbf{S}|=\frac{|S|N(N+1)}{2}=30$. It is easy to verify that $\text{Sparse}(\mathbf{S})$ is QI with respect to $\mathbf{CP}_{12}$, for example, by using the binary test \cite[Theorem~1]{furieri2019unified}. By direct computation of the Hessian through the Symbolic Math Toolbox Ver. 7.1 available in MATLAB \cite{MATLAB} we verify that  $J(\mathbf{K})$ is not convex on $\mathcal{K}$ \footnote{specifically we verify $\nabla^2\underline{f}(0)\not \succeq 0$, where $\underline{f}(\bm{\alpha})=f(M\bm{\alpha})$, the columns of $M$ are an orthonormal basis of $\mathcal{K}$ and $f(\text{vec}(\mathbf{K}))=J(\mathbf{K})$ }. Despite this non-convexity, we know by Theorem~\ref{th:GD} that the gradient-descent iteration (\ref{eq:iteration}) will converge to a global optimum of $\mathcal{P}_K$ for $t \rightarrow \infty$ thanks to the QI property.
\subsubsection{Numerical results} 
The gradient-descent iteration (\ref{eq:iteration}) was implemented in MATLAB with the stepsize being chosen according to the bisection algorithm of \cite[Proposition~5.5]{aragon2019nonlinear}. The iteration (\ref{eq:iteration}) was initialized from a variety of randomly selected initial distributed controllers. Specifically, for each entry $(i,j)$ such that $\mathbf{S}_{i,j}=1$ we selected the entry $(\mathbf{K}_0)_{i,j}$ uniformly at random in the interval $[-10,10]$, and set $(\mathbf{K}_0)_{i,j}=0$ otherwise. In all instances, we converged to a cost of $796.5627$ within up to $700$ iterations, with a run time of approximately $2$ seconds. The stopping criterion was selected as $\max |\nabla J(\mathbf{K}_t)\odot \mathbf{S}|<5\cdot 10^{-5}$.  To validate the global optimality result, we also solved the corresponding convex program (\ref{eq:convex_Q}) in $\mathbf{Q}$ with MOSEK \cite{mosek}, called through MATLAB via YALMIP \cite{YALMIP}, and obtained a minimum cost of $796.5627$.

\vspace{6pt}

At this point, it is natural to ask a follow-up question: is the QI/strong QI property \emph{necessary} to guarantee convergence of  gradient-descent to a globally optimal distributed controller? In the following section, we provide a negative answer.

\section{Unique Stationarity: Global Optimality Beyond QI}
\label{sec:US}
In this section, we consider general subspace constraints $\mathbf{K} \in \mathcal{K}$. We define the notion of unique stationarity (US) and show that it allows to step beyond the QI notion in obtaining global optimality certificates with first-order methods. We further provide initial results on verifying the US property in a tractable way.
 

\subsection{Unique stationarity generalizes QI}
We define unique stationarity of problem $\mathcal{P}_K$ as follows.
	\begin{definition}
		Consider problem $\mathcal{P}_K$ subject to a  subspace constraint $\mathbf{K}\in  \mathcal{K}$. We say that $\mathcal{P}_K$ is \emph{Uniquely Stationary} (US)  if and only if:
	\begin{equation}
	\nabla J(\overline{\mathbf{K}}) \in \mathcal{K}^\perp \implies \overline{\mathbf{K}}\in \arg \min_{\mathbf{K} \in \mathcal{K}}J(\mathbf{K})\,. \label{eq:US}
	\end{equation}
	\end{definition}
	
	First, it is easy to see that the class of US problem is at least as large as that of strongly QI problems.
	\begin{corollary}[Theorem~\ref{th:QI_uniquely_stationary}]
\label{co:th2}
Suppose that $\mathcal{K}$ is strongly \emph{QI}. Then, $\mathcal{P}_K$ is \emph{US}.
\end{corollary}
\begin{proof}
\emph{If $\mathbf{K}$ is a stationary point of $\mathcal{P}_{K}$ then $\Pi_{ \mathcal{K}}(\nabla J(\mathbf{K}))=0$. By Theorem~\ref{th:QI_uniquely_stationary} and Remark~\ref{re:generalization},  $\mathbf{K}$ is a global optimum. Hence, US as per (\ref{eq:US}) holds.}
\end{proof}
	
Second, we extend the global convergence result of Theorem~\ref{th:GD} from strongly QI to US problems. 
\begin{proposition}
Suppose that $\mathcal{P}_K$ is \emph{US}. Let $\mathbf{K}_0 \in \mathcal{K}$ and consider the iteration
\begin{equation}
\label{eq:projected_iteration}
\mathbf{K}_{t+1}=\mathbf{K}_t-\eta_t \Pi_{\mathcal{K}}\left(\nabla J(\mathbf{K}_t)\right)\,.
\end{equation}
 Then, $\mathbf{K}_t \in \mathcal{K}$ for every $t$ and there exists  $\eta_t$ for every $t$ such that $\lim_{t \rightarrow \infty} J(\mathbf{K}_t)=J^\star$, where $J^\star$ is the optimal value of problem $\mathcal{P}_K$.
\end{proposition}
\begin{proof}
\emph{The proof mirrors that of Theorem~\ref{th:GD} by selecting $M$ such that its columns are an orthonormal basis of $\mathcal{K}$ in proving that $(\ref{eq:projected_iteration})$ converges to a stationary point. Since $\mathcal{P}_K$ is US, every stationary point is optimal.}
\end{proof}	
In other words, every US problem can be solved to global optimality with projected gradient-descent. Third, we characterize a US problem that is neither strongly QI nor QI.

\subsubsection{Example US beyond QI}
Consider the system (\ref{eq:sys_disc}) and the cost function (\ref{eq:cost_K}) with
\begin{align*}
&A_t=\begin{bmatrix}
1&2\\
-1&-3
\end{bmatrix}\,, \quad B_t=I\,,~ C_t=I\,, ~\Sigma_0=I\,,\\
&M_t=I\,,~R_t=I\,,~\Sigma^w_{t}=0\,,~\mathbf{\Sigma}_v=I\,,~\forall t=0,1,2\,,
\end{align*}
and $\mu_0=\begin{bmatrix}0&1\end{bmatrix}^\mathsf{T}$ where we set a horizon of $N=2$. The controller $\mathbf{K}$ is subject to being in the form $\mathbf{K}=I_N \otimes K$ for some $K \in \mathbb{R}^{2 \times 2}$. In other words, we consider a static-controller $u_t=Ky_t$ in finite-horizon. Note that in the finite-horizon setup it is not necessary to require that $(A+BK)$ is Hurwitz, since the finite-horizon cost $J(\mathbf{K})$ is finite for every $\mathbf{K}$, as opposed to the infinite-horizon cases of \cite{bu2019lqr, feng2019exponential}. Additionally, we require that $K$ is decentralized, or equivalently $K \in \text{Sparse}(I_2)$. In summary, we enforce
\begin{equation*}
\mathbf{K} \in \mathcal{K}=\{\mathbf{K}=I_N \otimes \text{diag}(a,b),~a,b \in \mathbb{R}\}\,.
\end{equation*}
By computing $\mathbf{KCP}_{12}\mathbf{K}$ for a generic $\mathbf{K} \in \mathcal{K}$ it is easy to verify  that $\mathcal{K}$ is neither strongly QI or QI with respect to $\mathbf{CP}_{12}$. Hence, a convex program equivalent to $\mathcal{P}_K$ in the $\mathbf{Q}$ domain does not exist by Theorem~\ref{th:QI}. Nonetheless, we prove that $\mathcal{P}_K$ is US and can thus be solved to global optimality with gradient-descent.

\emph{Proof of US:}  For any $\mathbf{K} \in \mathcal{K}$ we verify
\begin{align*}
J(\mathbf{K})=\underline{f}(a,b)&=4a^4 + 8a^3 + 28a^2 + 18ab - 38a+ 6b^4 - 42b^3 + 149b^2 - 216b + 166\,.
\end{align*}

The expression above can be obtained by using the Symbolic Math Toolbox in MATLAB \cite{MATLAB}. The Hessian is
\begin{align*}
\nabla^2 \underline{f}(a,b)=\begin{bmatrix}
48a^2 + 48a + 56&18\\18&72b^2 - 252b + 298
\end{bmatrix}\,.
\end{align*}
We verify that  $48a^2 + 48a + 56 = 12(2a+1)^2+44>0$ for all $a\in \mathbb{R}$ and 
\begin{align*}
\det{\left(\nabla^2 \underline{f}(a,b)\right)}&=24(1+2a)^2\left(36\left(b-1.75\right)^2+38.75\right)+198(7-4b)^{2}+3086>0, \quad \forall a,b \in \mathbb{R}\,.
\end{align*}
It follows that $\nabla^2 \underline{f}(a,b)\succ 0$  for all  $a,b \in \mathbb{R}$, and hence $J(\mathbf{K})$ is convex on $\mathcal{K}$. We conclude that $\mathcal{P}_K$ is  US, despite not being QI. The globally optimal controller $\mathbf{K}^\star=I_N \otimes \text{diag}(0.2752,1.1354)$ is found on average in $11$ iterations of (\ref{eq:projected_iteration}) with the two free variables of $\mathbf{K}_0 \in \mathcal{K}$ randomly selected in $[-10,10]$, stepsize  as per \cite[Proposition~5.5]{aragon2019nonlinear} and stopping criterion $\max |\Pi_\mathcal{K}\left(\nabla J(\mathbf{K}_t)\right)|<5\cdot 10^{-5}$. 

Last, we summarize the main result of this section as follows. A corresponding visualization is reported in Figure~\ref{fig:scheme}.
	\begin{theorem}
	\label{th:US_QI}
	The class of \emph{US} problems is both
	\begin{enumerate}
	\item strictly larger than the class of strongly \emph{QI} problems, 
	\item  not included in the class of \emph{QI} problems.
	\end{enumerate}
	\end{theorem}
	



	\begin{proof}
\emph{Every strongly QI problem is US by Corollary~\ref{co:th2}.  We have shown an instance of a US problem which is  neither strongly QI or QI. This proves that the class of US problems is both strictly larger than strongly QI problems and not included in QI problems.}
	\end{proof}

		\begin{figure}[t]
	      \centering
	      	\includegraphics[width=0.5\textwidth]{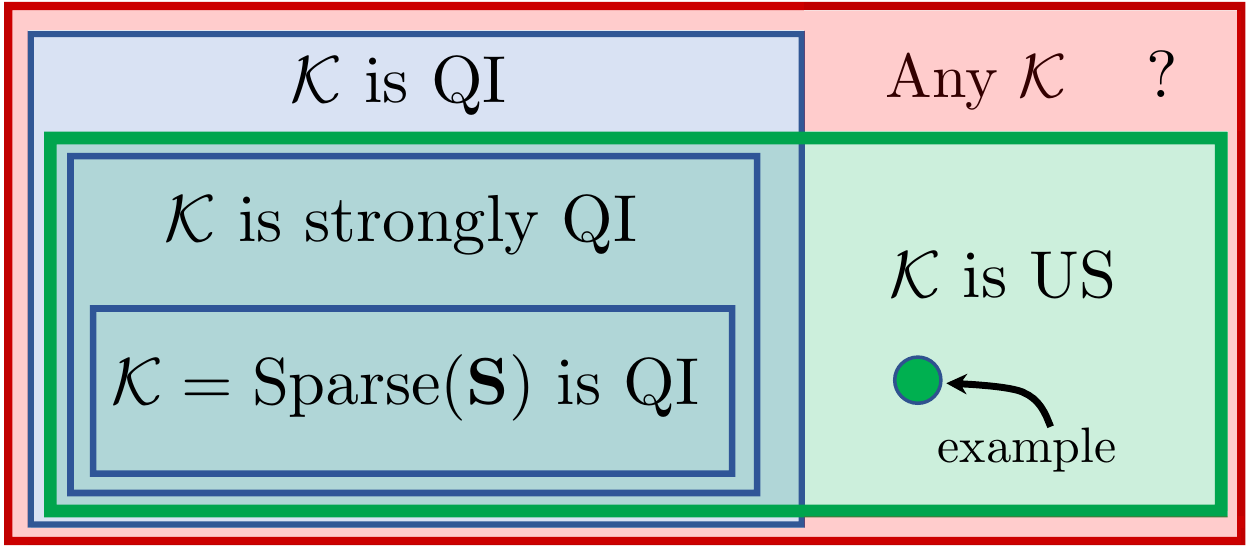}
    \caption{ Problems in the blue region can be solved with convex programming (QI problems). Problems in the green region can be solved with gradient-descent (US problems). The US region includes the case of QI sparsity constraints. The green circle stands for the explicit example we provided. Different methods are needed to solve problems in the red region; however, whether the red region contains any problem at all is an open question.  The question mark  indicates problems whose existence is yet to be verified.}
    \label{fig:scheme}
	\end{figure}
	
	We remark that the notion of US genuinely extends QI in terms of providing global optimality certificates for distributed control. This  might sound surprising at first.  To grasp this fact, notice that QI is method-specific, in the sense that it is only necessary for global optimality certificates when one uses  convex programming in the $\mathbf{Q}$ domain \cite{QIconvexity}. On the contrary, we have shown that $\mathcal{P}_K$ might be uniquely stationary and even convex in the original $\mathbf{K}$ coordinates despite being non-convex in the $\mathbf{Q}$ domain. This observation and  Corollary~\ref{co:th2} allow stepping beyond the QI limitations, from convexity in $\mathbf{Q}$ to unique stationarity in $\mathbf{K}$!
	

	\subsection{Tests for unique stationarity}
	Note that the US property, while having a theoretical interest,  might not be useful in practice in the form (\ref{eq:US}). This is because, in general, one can only prove (\ref{eq:US}) by knowing the set of global optima. For this reason, it is necessary to identify sufficient conditions for US. While noting that more general tests should be envisioned in future research, we provide our initial results. A first test of US given sparsity constraints follows naturally from Corollary~\ref{th:QI_uniquely_stationary}:
	\begin{corollary}
	Suppose that $\mathcal{K}=\text{\emph{Sparse}}(\mathbf{S})$ and Let $\bm{\Delta}=\text{\emph{Struct}}(\mathbf{CP}_{12})$. Then 
	\begin{equation*}
\mathbf{S}\bm{\Delta}\mathbf{S}\leq \mathbf{S} \implies \mathcal{P}_K\text{ is \emph{US}}\,.
\end{equation*}
	\end{corollary}
	\begin{proof}
	\emph{By \cite[Theorem~1]{furieri2019unified}, $\text{Sparse}(\mathbf{S})$ is strongly QI with respect to $\mathbf{CP}_{12}$ if and only if $\mathbf{S}\bm{\Delta}\mathbf{S}\leq \mathbf{S}$. Hence, $\mathcal{P}_K$ is US as a consequence of Theorem~\ref{th:QI_uniquely_stationary}.}
	\end{proof}
Notice that $\mathbf{S} \bm{\Delta} \mathbf{S}\leq \mathbf{S}$ is verified in polynomial time in $m$, $p$ and $N$. A second sufficient test for US beyond QI is to check whether $J(\mathbf{K})$ is convex on $\mathcal{K}$.

\begin{proposition}
Let $f:\mathbb{R}^{mpN^2} \rightarrow \mathbb{R}$ be such that $f(\text{\emph{vec}}(\mathbf{K}))=J(\mathbf{K})$ and $\underline{f}:\mathbb{R}^r\rightarrow \mathbb{R}$ be such that $\underline{f}(\bm{\alpha})=f(M\bm{\alpha})$ where the columns of $M$ are an orthonormal basis of $\mathcal{K}$ and $r$ is the dimension of $\mathcal{K}$. Let $\nabla^2_g\underline{f}(\bm{\alpha}) \in \mathbb{R}^{g \times g}$ denote the submatrix of $\nabla^2 \underline{f}(\bm{\alpha})$ obtained by removing its last $r-g$ rows and columns. Then
\begin{equation*}
\det\left(\nabla^2_g\underline{f}(\bm{\alpha}) \right)>0\quad \forall \bm{\alpha},~\forall g=1,\ldots,r \implies \mathcal{P}_K\text{ is US}\,.
\end{equation*}
\end{proposition}
\begin{proof}
\emph{By definition  $\underline{f}$ is convex if and only if $J$ is convex on $\mathcal{K}$. The function $\underline{f}$ is convex if and only $\nabla^2\underline{f}(\bm{\alpha})\succ 0$ for every $\bm{\alpha}$, or equivalently, the determinant of each principal minor of $\nabla^2\underline{f}(\bm{\alpha})$ is positive for every $\bm{\alpha}$. }
\end{proof}
Notice that $\det\left(\nabla^2_g\underline{f}(\bm{\alpha}) \right)$ is a polynomial for every $g$. Deciding positivity of multivariate polynomials is NP-hard in general, but it can be performed in finite time \cite{becker2000deciding}. When  $\det\left(\nabla^2_g\underline{f}(\bm{\alpha}) \right)$  is a Sum-of-squares (SOS) for every $g$, as in the example we provided, then the US property can be decided in polynomial time with standard techniques \cite{parrilo2000structured}.

	\section{Conclusions}
	\label{sec:conclusions}
	We have addressed convergence to a global optimum of first-order methods for the distributed discrete-time LQG problem in finite-horizon. If the strong QI property holds, a projected gradient-descent algorithm is guaranteed to converge to a global optimum. Moreover, we have characterized the class of uniquely stationary (US) problems, for which projected gradient-descent converges to a global optimum.  We have proved that the class of US problems is strictly larger than strongly QI problems and not included in QI problems. Our results indicate that first-order methods in the $\mathbf{K}$ domain are superior to convex programming in the $\mathbf{Q}$ domain in terms of generality of their  global optimality certificates and allow stepping beyond the long-standing QI limitation \cite{rotkowitz2006characterization}.  Additionally, first-order methods can be used  to learn globally optimal distributed controllers when the system and the cost function are unknown, as was recently shown in \cite{fazel2018global,gravell2019learning,hassan2019data} for the non-distributed case.  We envision that future work will discuss application of our methods to learning-based distributed control. 
	
This work initiates the research for novel classes of constrained and distributed control problems, for which a test of the US property  beyond QI and beyond testing convexity of $\mathcal{P}_K$ in the $\mathbf{K}$ domain is available. For instance, one could study under which conditions $J(\mathbf{K})$ is \emph{gradient dominated} on $\mathcal{K}$ \cite{polyak1963gradient}. 
	In the finite-horizon setting considered here, it is important to either confirm or disprove the existence of non-US problems, indicated by ``?'' in Figure~\ref{fig:scheme}. This insight  would further advance the comprehension of the mathematical challenges inherent to linear distributed control.   Last, it is  important to address the  infinite-horizon and continuous-time cases. In infinite-horizon,  \cite{feng2019exponential} provided  explicit examples of problems that are non-US due to the set of distributed static stabilizing controllers being disconnected; it is interesting to explore whether dynamic controllers  can mitigate this issue.
	
	\section*{Acknowledgements}
	We thank Tyler Summers and Ilnura Usmanova for useful discussions.

 \section*{Appendix}
 \subsection*{Derivation of the cost function $J(\mathbf{K})$}
Note that the cost (\ref{eq:cost}) is equivalent to
 \begin{equation}
 \label{eq:cost_compact}
 J(\mathbf{K})=\mathbb{E}_{\mathbf{w,v}}\left(\mathbf{x}^\mathsf{T}\mathbf{M}\mathbf{x}+\mathbf{u}^\mathsf{T}\mathbf{R}\mathbf{u}\right)\,.
 \end{equation} 
  Now consider the control input $\mathbf{u}=\mathbf{Ky}$. The closed-loop state, output and input trajectories are given in (\ref{eq:closed_loop_K}), where $\mathbf{x}$ and $\mathbf{u}$ are expressed as a function of $\mathbf{w}$ and $\mathbf{v}$. Substitute (\ref{eq:closed_loop_K}) into (\ref{eq:cost_compact}). By using the fact that for any matrix $\mathbf{X}$ we have
  \begin{align*}
  &\mathbb{E}_{\mathbf{w}}\left(\mathbf{w}^\mathsf{T}\mathbf{Xw}\right)=\text{Trace}(\mathbf{X}\mathbf{\Sigma}_w)+\bm{\mu}_w^\mathsf{T}\mathbf{X}\bm{\mu}_w\,,\\
  &\mathbb{E}_{\mathbf{v}}\left(\mathbf{v}^\mathsf{T}\mathbf{Xv}\right)=\text{Trace}(\mathbf{X}\mathbf{\Sigma}_v)\,,\quad\mathbb{E}_{\mathbf{w,v}}\left(\mathbf{w}^\mathsf{T}\mathbf{Xv}\right)=0\,,
  \end{align*}
  and remembering that $\norm{\mathbf{X}}_F^2=\text{Trace}(\mathbf{X}^\mathsf{T}\mathbf{X})$ we obtain the expression (\ref{eq:cost_K}).
\subsection*{Proof of Lemma~\ref{le:strictly_convex}}
\label{app:le:strictly_convex} 
$1)$ By using several relationships to compute derivatives with respect to matrices from \cite{petersen2008matrix}  and the fact that $\text{vec}(AXB)=(B^\mathsf{T} \otimes A) \text{vec}(X)$ we obtain that 
\begin{align*}
&\nabla \left(\text{vec}\left(\nabla \tilde{J}(\mathbf{Q})\right)\right)\\&=2\big[\underbrace{(\mathbf{\Sigma}_v+\mathbf{CP}_{11}\mathbf{\Sigma}_w\mathbf{P}_{11}^\mathsf{T}\mathbf{C}^\mathsf{T})}_{\succ 0}\otimes\underbrace{(\mathbf{R}+\mathbf{P}_{12}^\mathsf{T}\mathbf{M}\mathbf{P}_{12})}_{\succ 0}+\underbrace{\mathbf{CP}_{11}\bm{\mu}_w\bm{\mu}_w^\mathsf{T}\mathbf{P}_{11}^\mathsf{T}\mathbf{C}^\mathsf{T}}_{\succeq 0}\otimes \underbrace{(\mathbf{R}+\mathbf{P}_{12}^\mathsf{T}\mathbf{P}_{12 })}_{\succ 0}\big]\succ 0\,,\end{align*}
because $\mathbf{R,\Sigma}_v\succ 0$ and $\mathbf{M},\mathbf{\Sigma}_w \succeq 0$ by hypothesis.  It follows that $\tilde{J}(\mathbf{Q})$ is a quadratic form that is strictly convex.  The statements $2)$ and $3)$ follow 
from direct computation by exploiting the definition of the function $h$.


\subsection*{Proof of Lemma~\ref{le:bounded_level_sets}}
\label{app:le:bounded_level_sets}
 Since $\tilde{J}(\mathbf{Q})$ is strictly convex by Lemma~\ref{le:strictly_convex}, its sublevel set 
$\tilde{\mathcal{L}}:=\{\mathbf{Q}|~\tilde{J}(\mathbf{Q})\leq J(\mathbf{K}_0)\}$ is bounded for any $\mathbf{K}_0$ \cite[Ch. 9.1.2]{boyd2004convex}. Since $\tilde{J}(\mathbf{Q})=J(h(\mathbf{Q},\mathbf{CP}_{12}))$ for every $\mathbf{Q}$ we have $\mathcal{L}=h(\tilde{\mathcal{L}},\mathbf{CP}_{12})$. Now notice that
\begin{equation}
\label{eq:h_expanded}
h(\mathbf{Q},\mathbf{CP}_{12})=\sum_{i=0}^N(-1)^i(\mathbf{QCP}_{12})^i\mathbf{Q}\,,
\end{equation}
because each $m \times m$ block of $\mathbf{QCP}_{12}$ is the zero matrix by construction. Hence, every entry of matrix $h(\mathbf{Q},\mathbf{CP}_{12})$ is a multivariate polynomial in $\mathbf{Q}$, that is, a continuous function. We conclude that  $\mathcal{L}$ is bounded if and only if $\tilde{\mathcal{L}}$ is bounded. Since $\tilde{\mathcal{L}}$ is bounded for any $\mathbf{K}_0$, the result follows.

 \subsection*{Proof of Lemma~\ref{le:equivalence_stationary}}
 $\Leftarrow)$ In the interest of readability, in this proof we omit the second argument of the function $h(\cdot,\cdot)$, which is assumed to always be fixed to $\mathbf{CP}_{12}$. Assume that $\nabla J(\overline{\mathbf{K}}) \in \mathcal{K}^\perp$, but  $ \nabla \tilde{J}(\mathbf{\overline{\mathbf{Q}}})\not \in \mathcal{K}^\perp$.  Then, there exists $\tilde{\mathbf{Q}} \in \mathcal{K}$ with $\tilde{\mathbf{Q}} \neq 0$ with:
 \begin{equation*}
 \lim_{\epsilon \rightarrow 0} \frac{\tilde{J}(\overline{\mathbf{Q}}+\epsilon\tilde{\mathbf{Q}})-\tilde{J}(\overline{\mathbf{Q}})}{\epsilon}=k\neq 0\,.
 \end{equation*}
  Equivalently,  since $h(\cdot)$ is invertible,
  \begin{align}
& \lim_{\epsilon \rightarrow 0} \frac{\tilde{J}(h^{-1}(h(\overline{\mathbf{Q}}+\epsilon\tilde{\mathbf{Q}})))-\tilde{J}(\overline{\mathbf{Q}})}{\epsilon} \nonumber = \lim_{\epsilon \rightarrow 0} \frac{J(h(\overline{\mathbf{Q}}+\epsilon\tilde{\mathbf{Q}}))-J(\overline{\mathbf{K}})}{\epsilon}=k\neq 0\,.\label{eq:derivative_ginverse}
 \end{align}
 Now, using a first-order Taylor expansion we have
 \begin{align}
&h(\overline{\mathbf{Q}}+\epsilon\tilde{\mathbf{Q}})=(I+(\overline{\mathbf{Q}}+\epsilon \tilde{\mathbf{Q}})\mathbf{CP}_{12})^{-1}(\overline{\mathbf{Q}}+\epsilon \tilde{\mathbf{Q}})\\\nonumber
&=(I+\overline{\mathbf{Q}}\mathbf{CP}_{12}+\epsilon \tilde{\mathbf{Q}} \mathbf{CP}_{12})^{-1}(\overline{\mathbf{Q}}+\epsilon \tilde{\mathbf{Q}})\\\nonumber
&=\Big[ (I+\overline{\mathbf{Q}}\mathbf{CP}_{12})^{-1}-\epsilon(I+\overline{\mathbf{Q}}\mathbf{CP}_{12})^{-1}\tilde{\mathbf{Q}}\mathbf{CP}_{12}(I+\overline{\mathbf{Q}}\mathbf{CP}_{12})^{-1}+\mathcal{O}(\epsilon^2)\Big](\overline{\mathbf{Q}}+\epsilon \tilde{\mathbf{Q}})\\&=\overline{\mathbf{K}}+\epsilon \tilde{\mathbf{K}}+\mathcal{O}(\epsilon^2)\,,\nonumber
 \end{align}
 where
 \begin{equation}
 \label{eq:K_tilde}
 \tilde{\mathbf{K}}=(I+\overline{\mathbf{Q}}\mathbf{CP}_{12})^{-1}(\tilde{\mathbf{Q}}-\tilde{\mathbf{Q}}\mathbf{CP}_{12}(I+\overline{\mathbf{Q}}\mathbf{CP}_{12})^{-1} \overline{\mathbf{Q}})\,.
 \end{equation}
By (\ref{eq:h_expanded}) and by applying the strong QI property we deduce that $\tilde{\mathbf{K}} \in \mathcal{K}$. By substituting the above derivations into (\ref{eq:derivative_ginverse}):
 \begin{align*}
 &\lim_{\epsilon \rightarrow 0} \frac{J(\overline{\mathbf{K}}+\epsilon \tilde{\mathbf{K}}+\mathcal{O}(\epsilon^2))-J(\overline{\mathbf{K}})}{\epsilon}=\lim_{\epsilon \rightarrow 0} \frac{J(\overline{\mathbf{K}}+\epsilon \tilde{\mathbf{K}})-J(\overline{\mathbf{K}})}{\epsilon}=k\neq 0\,.
 \end{align*}
 Since $\tilde{\mathbf{K}} \in \mathcal{K}$ and is non-null due to $\tilde{\mathbf{Q}} \neq 0$, this contradicts $\nabla J(\overline{\mathbf{K}}) \in \mathcal{K}^\perp$.
 $\Rightarrow)$ can be proven analogously.

\subsection*{Proof of Lemma~\ref{le:unconstrained}}
Since $\text{Im}(M)=\mathcal{K}$, minimizing $\underline{f}$ on $\mathbb{R}^r$ is equivalent to minimizing $f$ on $\mathcal{K}$. Hence, the first point holds by definition of $\underline{f}$. For the second point, we have $\nabla \underline{f}(\bm{\alpha})=M^\mathsf{T}\nabla f(M\bm{\alpha})$ by the derivative chain rule. We deduce that $\nabla \underline{f}(\bm{\alpha})=0$ if and only if $\nabla f(M\bm{\alpha})\in \text{Ker}(M^\mathsf{T})=\text{Im}(M)^\perp=\mathcal{K}^\perp\,.$
 
%
%

	\bibliographystyle{IEEEtran}

	\bibliography{IEEEabrv,references2}

\begin{thebibliography}{10}
\providecommand{\url}[1]{#1}
\csname url@samestyle\endcsname
\providecommand{\newblock}{\relax}
\providecommand{\bibinfo}[2]{#2}
\providecommand{\BIBentrySTDinterwordspacing}{\spaceskip=0pt\relax}
\providecommand{\BIBentryALTinterwordstretchfactor}{4}
\providecommand{\BIBentryALTinterwordspacing}{\spaceskip=\fontdimen2\font plus
\BIBentryALTinterwordstretchfactor\fontdimen3\font minus
  \fontdimen4\font\relax}
\providecommand{\BIBforeignlanguage}[2]{{%
\expandafter\ifx\csname l@#1\endcsname\relax
\typeout{** WARNING: IEEEtran.bst: No hyphenation pattern has been}%
\typeout{** loaded for the language `#1'. Using the pattern for}%
\typeout{** the default language instead.}%
\else
\language=\csname l@#1\endcsname
\fi
#2}}
\providecommand{\BIBdecl}{\relax}
\BIBdecl

\bibitem{Witsenhausen}
H.~S. Witsenhausen, ``A counterexample in stochastic optimum control,''
  \emph{SIAM Journal on Control}, vol.~6, no.~1, pp. 131--147, 1968.

\bibitem{papadimitriou1986intractable}
C.~H. Papadimitriou and J.~Tsitsiklis, ``Intractable problems in control
  theory,'' \emph{SIAM j. on contr. and opt.}, vol.~24, no.~4, pp. 639--654,
  1986.

\bibitem{youla1976modern}
D.~Youla, H.~Jabr, and J.~Bongiorno, ``{Modern Wiener-Hopf design of optimal
  controllers--Part II: The multivariable case},'' \emph{IEEE Trans. on Aut.
  Contr.}, vol.~21, no.~3, pp. 319--338, 1976.

\bibitem{rotkowitz2006characterization}
M.~Rotkowitz and S.~Lall, ``A characterization of convex problems in
  decentralized control,'' \emph{IEEE Trans. on Aut. Contr.}, vol.~51, no.~2,
  pp. 274--286, 2006.

\bibitem{QIconvexity}
L.~Lessard and S.~Lall, ``Quadratic invariance is necessary and sufficient for
  convexity,'' in \emph{American Control Conference (ACC), 2011}.\hskip 1em
  plus 0.5em minus 0.4em\relax IEEE, 2011, pp. 5360--5362.

\bibitem{furieri2019sparsity}
L.~Furieri, Y.~Zheng, A.~Papachristodoulou, and M.~Kamgarpour, ``Sparsity
  invariance for convex design of distributed controllers,'' \emph{arXiv
  preprint arXiv:1906.06777}, 2019.

\bibitem{wang2019system}
Y.-S. Wang, N.~Matni, and J.~C. Doyle, ``A system level approach to controller
  synthesis,'' \emph{IEEE Trans. on Aut. Contr.}, 2019.

\bibitem{SDP}
G.~Fazelnia, R.~Madani, A.~Kalbat, and J.~Lavaei, ``Convex relaxation for
  optimal distributed control problems,'' \emph{IEEE Trans. on Aut. Contr.},
  vol.~62, no.~1, pp. 206--221, 2017.

\bibitem{wang2018convex}
Y.~Wang, J.~A. Lopez, and M.~Sznaier, ``Convex optimization approaches to
  information structured decentralized control,'' \emph{IEEE Trans. on Aut.
  Contr.}, 2018.

\bibitem{dvijotham2015convex}
K.~Dvijotham, E.~Todorov, and M.~Fazel, ``Convex structured controller design
  in finite horizon,'' \emph{IEEE Trans. on Contr. of Netw. Sys.}, vol.~2,
  no.~1, pp. 1--10, 2015.

\bibitem{lin2011augmented}
F.~Lin, M.~Fardad, and M.~R. Jovanovic, ``Augmented {Lagrangian} approach to
  design of structured optimal state feedback gains,'' \emph{IEEE Trans. on
  Aut. Contr.}, vol.~56, no.~12, pp. 2923--2929, 2011.

\bibitem{fazel2018global}
M.~Fazel, R.~Ge, S.~M. Kakade, and M.~Mesbahi, ``Global convergence of policy
  gradient methods for the linear quadratic regulator,'' \emph{arXiv preprint
  arXiv:1801.05039}, 2018.

\bibitem{gravell2019sparse}
B.~Gravell, Y.~Guo, and T.~Summers, ``Sparse optimal control of networks with
  multiplicative noise via policy gradient,'' \emph{arXiv preprint
  arXiv:1905.13547}, 2019.

\bibitem{gravell2019learning}
B.~Gravell, P.~Mohajerin~Esfahani, and T.~Summers, ``Global convergence of
  policy gradient methods for the linear quadratic regulator,'' \emph{arXiv
  preprint arXiv:1905.13547}, 2019.

\bibitem{bu2019lqr}
J.~Bu, A.~Mesbahi, M.~Fazel, and M.~Mesbahi, ``{LQR} through the lens of first
  order methods: Discrete-time case,'' \emph{arXiv preprint arXiv:1907.08921},
  2019.

\bibitem{Hesameddin2019global}
M.~Hesameddin, Z.~Armin, S.~Mahdi, and M.~Jovanovic, ``Global exponential
  convergence of gradient methods over the nonconvex landscape of the linear
  quadratic regulator,'' \emph{Conference on Decision and Control (CDC)}, 2019.

\bibitem{maartensson2009gradient}
K.~M{\aa}rtensson and A.~Rantzer, ``Gradient methods for iterative distributed
  control synthesis,'' in \emph{Conference on Decision and Control
  (CDC)}.\hskip 1em plus 0.5em minus 0.4em\relax IEEE, 2009, pp. 549--554.

\bibitem{hassan2019data}
S.~Hassan-Moghaddam, M.~R. Jovanovi{\'c}, and S.~Meyn, ``Data-driven proximal
  algorithms for the design of structured optimal feedback gains,'' in
  \emph{American Control Conference (ACC)}.\hskip 1em plus 0.5em minus
  0.4em\relax IEEE, 2019, pp. 5846--5850.

\bibitem{feng2019exponential}
H.~Feng and J.~Lavaei, ``On the exponential number of connected components for
  the feasible set of optimal decentralized control problems,'' in
  \emph{American Control Conference (ACC)}, 2019, p.~8.

\bibitem{furieri2019unified}
L.~Furieri and M.~Kamgarpour, ``Unified approach to convex robust distributed
  control given arbitrary information structures,'' \emph{IEEE Trans. on Aut.
  Contr.}, 2019.

\bibitem{colombino2017mutually}
M.~Colombino, R.~Smith, and T.~Summers, ``Mutually quadratically invariant
  information structures in two-team stochastic dynamic games,'' \emph{IEEE
  Trans. on Aut. Contr.}, vol.~63, no.~7, pp. 2256--2263, 2017.

\bibitem{nocedal2006numerical}
J.~Nocedal and S.~Wright, \emph{Numerical optimization}.\hskip 1em plus 0.5em
  minus 0.4em\relax Springer Science \& Business Media, 2006.

\bibitem{aragon2019nonlinear}
F.~J. Arag{\'o}n, M.~A. Goberna, M.~A. L{\'o}pez, and M.~M. Rodr{\'\i}guez,
  \emph{Nonlinear optimization}.\hskip 1em plus 0.5em minus 0.4em\relax
  Springer, 2019.

\bibitem{MATLAB}
\emph{MATLAB R2016b and Symbolic Math Toolbox Version 7.1}.\hskip 1em plus
  0.5em minus 0.4em\relax Natick, Massachusetts: The MathWorks Inc., 2016.

\bibitem{mosek}
{MOSEK Aps}, ``The {MOSEK} optimization toolbox for {MATLAB} manual. {V}ersion
  8.1.'' 2017.

\bibitem{YALMIP}
J.~L{\"{o}}fberg, ``{YALMIP : A Toolbox for Modeling and Optimization in
  MATLAB},'' in \emph{In Proc. of the CACSD Conf.}, Taipei, Taiwan, 2004.

\bibitem{becker2000deciding}
E.~Becker, V.~Powers, and T.~Wormann, ``Deciding positivity of real
  polynomials,'' \emph{Contemporary Mathematics}, vol. 253, pp. 19--24, 2000.

\bibitem{parrilo2000structured}
P.~A. Parrilo, ``Structured semidefinite programs and semialgebraic geometry
  methods in robustness and optimization,'' Ph.D. dissertation, California
  Institute of Technology, 2000.

\bibitem{polyak1963gradient}
B.~T. Polyak, ``Gradient methods for the minimisation of functionals,''
  \emph{USSR Computational Mathematics and Mathematical Physics}, vol.~3,
  no.~4, pp. 864--878, 1963.

\bibitem{petersen2008matrix}
K.~B. Petersen, M.~S. Pedersen \emph{et~al.}, ``The matrix cookbook,''
  \emph{Technical University of Denmark}, vol.~7, no.~15, p. 510, 2008.

\bibitem{boyd2004convex}
S.~Boyd and L.~Vandenberghe, \emph{Convex optimization}.\hskip 1em plus 0.5em
  minus 0.4em\relax Cambridge university press, 2004.

\end{thebibliography}

	\newpage

\end{document}